\documentclass[runningheads,a4paper]{llncs}
\usepackage{breakurl} 
\usepackage{amssymb}
\usepackage{graphicx}
\usepackage{mya4}
\usepackage{amssymb}
\usepackage{amsmath}
\usepackage{stmaryrd}
\usepackage{prooftree}
\usepackage{paralist}
\usepackage[matrix,frame,arrow]{xypic}
\usepackage{url}
\usepackage{datetime}
\usepackage{tikz}
\usepackage{graphicx}
\usepackage{epsfig}
\usetikzlibrary{trees,arrows,positioning,calc,decorations.pathmorphing}
\usepackage[all]{xy}

\newtheorem{Theorem}{Theorem}
\newtheorem{Definition}{Definition}
\newcommand{\mktype}[1]{\mathsf{#1}}
\newcommand{\Int}{\mktype{Int}}

\newcommand{\Qbit}{\mktype{Qbit}}
\newcommand{\Qdit}{\mktype{Qdit}}

\newcommand{\NS}{\mktype{NS}}

\newcommand{\Val}{\mktype{Val}}

\newcommand{\Op}[1]{\mktype{Op}(#1)}

\newcommand{\mkterm}[1]{\mathsf{#1}}

\newcommand{\unit}{\mkterm{unit}}
\newcommand{\ctxt}[2]{{#1}[#2]}
\newcommand{\bnf}{::=}
\newcommand{\alt}{~|~}

\newcommand{\Prob}[1]{\boxplus_{#1}}

\newcommand{\cnfig}[3]{({#1};{#2};{#3})}

\newcommand{\ptrns}[3]{{#1}\stackrel{#2}{\longrightarrow}{#3}}

\newcommand{\mkRrule}[1]{\mbox{\textsc{R-#1}}}

\newcommand{\Rplus}{\mkRrule{Plus}}

\newcommand{\Rmeasure}{\mkRrule{Measure}}

\newcommand{\Rcontext}{\mkRrule{Context}}
\newcommand{\Rtrans}{\mkRrule{Trans}}

\newcommand{\typed}[2]{{#1}\mathrel{\!:\!}{#2}}

\newcommand{\ket}[1]{|#1\rangle}
\newcommand{\nil}{\mathbf{0}}
\renewcommand{\parallel}{\mathbin{\mid}}
\newcommand{\inp}[2]{{#1}?{[#2]}}
\newcommand{\outp}[2]{{#1}!{[#2]}}

\renewcommand{\vec}[1]{\widetilde{#1}}
\newcommand{\new}{\mathsf{new}\ }

\newcommand{\qdit}{\mathsf{qdit}\ }

\newcommand{\chant}[1]{\widehat{~}[#1]}
\newcommand{\tid}[2]{{#1}\mathrel{\!:\!}{#2}}

\newcommand{\qgate}[1]{\mathsf{#1}}
\newcommand{\pname}[1]{\mathit{#1}}
\newcommand{\action}[1]{\{#1\}}

\newcommand{\trans}{\mathbin{*\!\!=}}
\newcommand{\sep}{\,.\,}

\newcommand{\measure}{\mathsf{measure}\ }

\newcommand{\qstore}[2]{[#1 \mapsto #2]}

\newcommand{\transition}[1]{\stackrel{#1}{\longrightarrow}}
\newcommand{\transitione}{\longrightarrow_e}
\newcommand{\transitionv}{\longrightarrow_v}

\newcommand{\ptrans}[1]{\stackrel{#1}\rightsquigarrow}
\newcommand{\weaktrans}[1]{\stackrel{#1}{\Longrightarrow}}
\newcommand{\opttrans}[1]{\stackrel{#1}{\longrightarrow}^{+}}

\newcommand{\cfgdistsum}{\oplus}

\newcommand{\Dist}[2]{\cfgdistsum_{#1}~#2~}

\newcommand{\bra}[1]{\langle#1|}

\newcommand{\ketbra}[2]{\ket{#1}\bra{#2}}

\newcommand{\rhoe}{\rho_E}

\newcommand{\gX}{\qgate{X}}

\newcommand{\gZ}{\qgate{Z}}
\newcommand{\gH}{\qgate{H}}
\newcommand{\gRC}{\qgate{R}_{c}}
\newcommand{\gLC}{\qgate{L}_{c}}

\newcommand{\ltrm}[3]{\lambda{#1}\bullet{#2}; {#3}}

\newcommand{\states}{\ensuremath{\mathcal{S}}}
\newcommand{\nstates}{\ensuremath{\mathcal{S}_n}}
\newcommand{\pstates}{\ensuremath{\mathcal{S}_p}}

\newcommand{\pbsim}{\leftrightarroweq}
\newcommand{\fpbsim}{\leftrightarroweq^c}

\renewcommand{\vec}[1]{\widetilde{#1}}
\newcommand{\trace}{\mathrm{tr}}

\newcommand{\pqwire}{\mathit{QWire}}
\newcommand{\pteleport}{\mathit{Teleport}}

\newcommand{\palice}{\pname{Alice}}
\newcommand{\pbob}{\pname{Bob}}

\input{Qcircuit}
\begin{document}

\mainmatter  

\title{Formal verification of higher dimensional quantum protocols}

\titlerunning{Verification of higher dimensional quantum protocols using quantum formal methods}
\authorrunning{I. V. Puthoor}

\author{Ittoop Vergheese Puthoor}

\institute{School of Computing Science,
\\ Newcastle University, UK
}

%
%

\toctitle{Lecture Notes in Computer Science}
\tocauthor{Authors' Instructions}
\maketitle

\begin{abstract}
Formal methods have been a successful approach for modelling and verifying the correctness of complex technologies like microprocessor chip design, biological systems and others. This is the main motivation of developing quantum formal techniques which is to describe and analyse quantum information processing systems. Our previous work demonstrates the possibility of using a quantum process calculus called Communicating Quantum Processes (CQP) to model and describe higher dimensional quantum systems. By developing the theory to generalise the fundamental gates and Bell states, we have modelled quantum qudit protocols like teleportation and superdense coding in CQP. In this paper, we demonstrate the use of CQP to analyse higher dimensional quantum protocols. The main idea is to define two processes, one modelling the real protocol and the other expressing a specification, and prove that they are behaviourally equivalent. This is a work-in-progress and we present our preliminary results in extending the theory of behavioural equivalence in CQP to verify higher dimensional quantum protocols using qudits.

\end{abstract}

\section{Introduction}
\label{sec-intro}
\label{sec:intro}

Quantum technologies is a rapidly growing field that offers tremendous potential in applications such as computing, communications, imaging and metrology. Quantum computing promises to offer high improvements in the performance of certain computations by exploiting the inherent parallelism that is achieved using the principles of quantum physics. For example, Shor's algorithm~\cite{Shor1994} is one of the few known quantum algorithms that demonstrates polynomial speedup compared to any known classical algorithm for the factorisation of integers. On the other hand we have quantum cryptographic systems which provide security that cannot be broken even if there are systems that have unlimited computing power. These computing and cryptographic systems are complex technologies and it is essential to ensure that these systems are reliable. Hence, there is need now to develop tools and techniques from \emph{formal methods} for modelling and verifying the correctness of these complex systems.

\emph{Formal methods} is widely regarded as a successful approach in the modelling and verification of classical systems. By using the rigorously specified mathematical theories and tools from \emph{formal methods}, one can build and test systems in order to ensure correct behaviour. The success of formal approaches in analysing classical systems provided the main motivation to extend these techniques and tools to specify and verify the correctness of practical quantum technologies. Such techniques are referred to as quantum formal methods. Based on process calculus, this line of research in quantum formal methods, known as \emph{quantum process calculus}, helps to describe and analyse the behaviour of systems that combine both quantum and classical elements. 

In this paper, we use an approach based on a quantum process calculus called Communicating Quantum Processes (CQP), developed by Gay and Nagarajan \cite{Gay2005}. 
The concept of \emph{behavioural equivalence} between processes is a congruence has been established in CQP \cite{DavidsonThesis}, and has been used to analyse quantum protocols such as teleportation, superdense coding and quantum error correction \cite{Davidson2011,PuthoorThesis}. Furthermore, there has been approaches in CQP to analyse realistic quantum information processing systems such as linear optical quantum computing \cite{Arnold2013,Arnold2014}. 

Previous work has shown that the syntax and semantics of CQP is not just limited to model systems that involves qubits but also can describe higher dimensional quantum systems i.e. qudits (a quantum system with $d$-dimensional Hilbert space) \cite{Gay2013}. This paper is a work in progress and presents the preliminary results for extending the theory of behavioural equivalence in CQP to analyse higher dimensional quantum systems. 

.

\section{Background}
\label{sec-prelim}
\label{sec:prelim}
 A \emph{qudit} is a physical system that is associated with a complex \emph{Hilbert space} $\mathbb{H}$, which is a $d$-dimensional vector space over the complex numbers, $\mathbb{C}$ with a basis denoted  by $\{\ket{0},\ket{1},\sep\sep\sep,\ket{d-1}\}$. Generally, we write the state of a qudit as 
 \begin{equation}
 \ket{\psi} = \sum^{d-1}_{i=0}\alpha_{i}\ket{i}
 \end{equation}
where $\alpha_{i} \in \mathbb{C}$ and $\Sigma^{d-1}_{i=0} \mid\alpha_{i}\mid^{2} = 1$. 
 Then, for a two-qudit quantum system $\mathbb{H_{A}} \otimes \mathbb{H_{B}}$, we recall the expressions for the quantum gate operators from  \cite{Gay2013} as: the generalised CNOT Right-Shift gate ($R_{C}$) that has control qudit $\ket{\psi} \in \mathbb{H_{A}}$  and target qudit $\ket{\phi} \in \mathbb{H_{B}}$. Then, for the set of standard basis states $\ket{m} \otimes \ket{n}$ of $\mathbb{H_{A}} \otimes \mathbb{H_{B}}$, we say, $R_{C}\ket{m} \otimes \ket{n} = \ket{m} \otimes \ket{n \oplus m}$. Likewise, the CNOT Left-shift gate ($L_{C}$) is given by ${L}_{C}\ket{m} \otimes \ket{n} = \ket{m} \otimes \ket{n \ominus m}$. The generalised Pauli operators are given by $X^{j}\ket{m} = \ket{m \oplus j}$ and $ Z^{k}\ket{m} = exp(2\pi ikm/d) \ket{m} = \omega^{km}\ket{m}$ where $\omega \equiv exp(2\pi i/d)$. Another essential unitary operator useful for manipulating qudits is the Hadamard operator denoted as $H$ and defined by $ H\ket{j} =  \frac{1}{\sqrt{d}}\sum_{m=0}^{d-1}\omega^{-jm}\ket{m}$. Together with operators $R_{C}$ and $H$ acting on the quantum system, we obtain the maximal entangled state, which is referred to as the generalised Bell state, defined as $\ket{\Psi^{nm}}_{AB} = \frac{1}{\sqrt{d}}\sum_{j=0}^{d-1}\omega^{-jn}\ket{j}_{A} \otimes \ket{j \oplus m}_{B}$. Please see \cite{Gay2013} for detailed information on the generalised quantum gates.

\section{Modelling quantum protocols in process calculus}
\label{sec-CQP}
\label{sec:CQP}
Communicating Quantum Processes (CQP) \cite{Gay2005} is a quantum process calculus that is based on the $\pi$-calculus  \cite{Milner1999,Milner1992} with primitives for quantum information. This was developed for formally modelling and analysing systems that combine both quantum and classical communicating and computing systems. General concept is that a system is considered to be made up of independent components or \emph{processes}. The \emph{processes} can communicate by sending and receiving data along \emph{channels} and these data are qubits or classical values. A distinctive feature of CQP is its static type system \cite{Gay2006a}, the purpose of which is to classify classical and quantum data and also to enforce the no-cloning property of quantum information. Due to space constraints, the syntax and semantics are presented in the appendix for reference. 

\subsection{Teleportation}

\begin{figure}
 \Qcircuit @C=1em @R = 2.0em {
    & & & & &  \lstick{x = \ket{\psi}}       &     \qw  &  \qw  &  \qw  & \qw & \ctrl{1} & \qw & \gate{H} & \qw & \qw & \qw & \meter & \control \cw \cwx[2]  &\\ 
    & & & & &   \lstick{z = \ket{0}}       &     \gate{H} & \ctrl{1} & \qw  & \qw & \gate{L_C} & \qw & \qw & \qw & \meter & \control \cw \cwx[1] &\\   
     & & & & &  \lstick{y = \ket{0}}      &     \qw &  \gate{R_C} & \qw & \qw & \qw & \qw & \qw & \qw & \qw &  \gate{X^{-M_1}} & \qw & \gate{Z^{M_2}} & \qw  & \rstick{\ket{\psi}}
}
\caption{\label{fig:Teleportation}Teleportation of qudit $\ket{\psi}$}
\end{figure} 

Quantum teleportation \cite{Barnett2009} is a protocol, which allows two users who share an entangled pair of qudits, to exchange an unknown quantum state by communicating only two classical values depending on the dimension $d$ of the system. The quantum circuit model of the protocol for qudits is shown in Figure~\ref{fig:Teleportation}. Note that the above quantum circuit model doesn't provide the complete description of the protocol. For example, it doesn't provide information on the number of users in the protocol and the communications between the users. The benefit of using the CQP language is that it captures these key informations and provide a clear and formal descriptions. The CQP definition of the teleportation protocol is given as
\[
\begin{array}{ll}
\multicolumn{2}{l}{\pteleport = (\qdit y,z)(\action{z\trans\gH}\sep\action{z,y\trans\gRC}\sep(\new e)(\palice(c,e) \parallel \pbob(e,d)))}
\end{array}
\]
This consists of two processes:  $\palice$ and $\pbob$,  we say the sender is $\palice$ and the receiver is $\pbob$.  $\palice$ possesses the qudit labelled $x$ which is in some unknown state $\ket{\psi}$; this is the qudit to be teleported. Qudits $y$ and $z$ are an entangled pair, which is generated by applying a Hadamard and CNOT- Right Shift gate to the qudits. Then qudit $z$ is given to $\palice$ and qudit $y$ is given to $\pbob$.The CQP definitions of $\palice$ and $\pbob$ are as follows,
\[
\begin{array}{ll}
\multicolumn{2}{l}{\palice(\typed{c}{\chant{\Qdit}},\typed{e}{\chant{\Val,\Val}}) = \inp{c}{\typed{x}{\Qdit}}\sep\action{x,z\trans\gLC}\sep\action{x\trans\gH}\sep\outp{e}{\measure z,\measure x}\sep\nil}
\end{array}
\]
\[
\begin{array}{ll}
 \multicolumn{2}{l}{\pbob(\typed{e}{\chant{\Val,\Val}},\typed{d}{\chant{\Qdit}}) = \inp{e}{\typed{M_{1}}{\Val},\typed{M_{2}}{\Val}}\sep\action{y\trans\gX^{-M_{1}}}\sep\action{y\trans\gZ^{M_{2}}}\sep\outp{d}{y}\sep\nil}
\end{array}
\]
The execution of the teleportation protocol are discussed in our previous work \cite{Gay2013}. The focus of this work is to analyse the protocol by extending the theory of equivalence in CQP to analyse the qudit teleportation protocol. We present our preliminary results in the following sections.
\section{Behavioural Equivalence of CQP Processes}
\label{sec-equivalence}
\label{sec:equivalence}
In this section, we present our preliminary work or results in applying the theory of equivalence \cite{DavidsonThesis} in CQP to verify higher dimensional quantum protocols like qudit teleportation. The theory was developed for qubits and in this paper we briefly present the initial investigations in extending the theory of equivalence to analyse qudit protocols. Although, it is assumed to be a straightforward task but the definitions and theorems needs to be checked and is still a work in progress. Previous section demonstrated the CQP model of qudit teleportation ($\pname{Teleport}$) and the execution of the protocol is discussed in detail in our previous work \cite{Gay2013}. In order to analyse $\pname{Teleport}$, we formally define a specification process $\pname{QWire}$ that describes the high-level observational behaviour of $\pname{Teleport}$. Verification is then achieved by proving that the two processes are bisimilar.

The concept of behavioural equivalence in CQP is a congruence and is a form of \emph{probabilistic branching bisimilarity} \cite{Trcka2008}. This is mainly adapted due to the probabilistic behaviour that arises from a quantum measurement. It is important to note that the reduced density matrices of the transmitted quantum information, in our case qudits, are required to be equal when considering the matching of input or output transitions. The following definitions are an extension from Davidson's
thesis \cite{DavidsonThesis} that is applicable in general to qudits.

\textbf{Notation:} Let $\opttrans{\tau}$ denote zero or one $\tau$ transitions; let $\weaktrans{ }$ denote zero or more $\tau$ transitions; and let $\weaktrans{\alpha}$ be equivalent to $\weaktrans{ }
\transition{\alpha} \weaktrans{ }$. We write $\vec{q}$ for a list of qudit names, and similarly for other lists.

\begin{Definition}[Density Matrix of Configurations]
  \label{def:density_matrix_configs}
  Let $\sigma_{ij} = \qstore{\vec{x}}{\ket{\psi_{ij}}}$ and $\vec{y}
  \subseteq \vec{x}$ and $t_{ij} = (\sigma_{ij}; \omega;
  \ltrm{\vec{w}}{P}{\vec{v_{ij}}})$ and $t = \Dist{ij}{g_{ij}} t_{ij}$. Then
\[
\begin{array}{llcll}
1. & \rho(\sigma_{ij}) = \ketbra{\psi_{ij}}{\psi_{ij}} & ~~~~~~ & 4. &
\rho^{\vec{y}}(t_{ij}) = \rho^{\vec{y}}(\sigma_{ij})  \\
2. & \rho^{\vec{y}}(\sigma_{ij}) =
\trace_{\vec{x}\setminus\vec{y}}(\ketbra{\psi_{ij}}{\psi_{ij}}) & & 5. &
\rho(t) = \sum_{ij} g_{ij} \rho(t_{ij}) \\
3. & \rho(t_{ij}) = \rho(\sigma_{ij}) & & 6. & \rho^{\vec{y}}(t) = \sum_{ij} g_{ij} \rho^{\vec{y}}(t_{ij})
\end{array}
\]
\end{Definition}

We also introduce the notation $\rhoe$ to denote the reduced density matrix of the \emph{environment} qudits. Formally, if $t = (\qstore{\vec{x}}{\ket{\psi}}; \vec{y}; P)$ then $\rhoe(t) = \rho^{\vec{r}}(t)$ where $\vec{r} = \vec{x} \setminus \vec{y}$. The definition of $\rhoe$ is extended to mixed configurations in the same manner as $\rho$. The probabilistic function $\mu: \states \times \states \rightarrow
[0,1]$ is defined in the style of \cite{Trcka2008}. It allows non-deterministic transitions to be treated as transitions with probability $1$, which is necessary when calculating the total probability of reaching a terminal state. $\mu(t,u) = \delta$ if $t \ptrans{\delta} u$; $\mu(t,u) = 1$ if $t = u$ and
$t \in \nstates$; $\mu(t,u) = 0$ otherwise.

\begin{Definition}[Probabilistic Branching Bisimulation]
  \label{def:pbb}
  An equivalence relation $\mathcal{R}$
  on configurations is a \emph{probabilistic branching bisimulation} on configurations if whenever $(t,u) \in \mathcal{R}$ the following conditions are satisfied.
  \begin{enumerate}[I.]
	    \item If $t \in \nstates$ and $t \transition{\tau} t'$ 
	      then $\exists u', u''$ such that $u \weaktrans{ } u'
              \opttrans{\tau} u''$ with $(t, u') \in \mathcal{R}$ and $(t', u'') \in \mathcal{R}$.
	    \item If $t \transition{\outp{c}{V,\vec{X}_1}} t'$ where $t' = \Prob{j \in \{1\dots m\}}{p_j} t_j'$ and $V = \{\vec{v}_1,\dots,\vec{v}_m\}$ and $\vec{X}_1$ is either $\vec{q}_1$ or $\vec{s}_1$ then $\exists u', u''$ such that $u \weaktrans{ } u' \transition{\outp{c}{V,\vec{X}_2}} u''$ with
	      \begin{enumerate}[a)]
		\item $(t, u') \in \mathcal{R}$,
		\item $u'' = \Prob{j \in \{1\dots m\}}{p_j} u_j''$,
                \item for each $j \in \{1,\dots,m\}$, $\rhoe(t_j') = \rhoe(u_j'')$.
		\item for each $j \in \{1,\dots,m\}$, $(t_j', u_j'') \in \mathcal{R}$.
	      \end{enumerate}
        \item If $t \transition{\inp{c}{\vec{v}}} t'$ then $\exists
          u', u''$ such that $u \weaktrans{ } u'
          \transition{\inp{c}{\vec{v}}} u''$ with $(t, u') \in
          \mathcal{R}$ and $(t', u'') \in \mathcal{R}$.
	    \item If $s \in \pstates$ then $\mu(t, D) = \mu(u, D)$ for all classes $D \in \mathcal{T}/\mathcal{R}$.
    \end{enumerate}
\end{Definition}
This relation follows the standard definition of branching bisimulation \cite{Glabbeek1996} with additional conditions for probabilistic configurations and matching quantum information. In
condition II we require that the distinct set of values $V$ must match and although the names ($\vec{X}_1$ and $\vec{X}_2$) need not be identical which is the qudit names ($\vec{q}_1$ and $\vec{q}_2$), their respective reduced density matrices ($\rho^{\vec{X}_1}(t)$ and $\rho^{\vec{X}_2}(u')$) must.
Condition IV provides the matching on probabilistic configurations following the approach of \cite{Trcka2008}. It is necessary to include the latter condition to ensure that the
probabilities are paired with their respective configurations. This leads to the following definitions. The essential definitions are presented in this paper and more details will be presented in the future work.

\begin{Definition}[Probabilistic Branching Bisimilarity]
Configurations $t$ and $u$ are \emph{probabilistic branching bisimilar}, denoted $t \pbsim u$, if there exists a probabilistic branching bisimulation $\mathcal{R}$ such that $(t,u) \in \mathcal{R}$.
\end{Definition}

\begin{Definition}[Probabilistic Branching Bisimilarity of Processes]
Processes $P$ and $Q$ are \emph{probabilistic branching bisimilar}, denoted $P \pbsim Q$, if and only if for all $\sigma$, $(\sigma; \emptyset; P) \pbsim (\sigma; \emptyset; Q)$.
\end{Definition}

\begin{Definition}[Full probabilistic branching bisimilarity]
  Processes $P$ and $Q$ are \emph{full probabilistic branching bisimilar}, denoted $P \fpbsim Q$, if for all substitutions $\kappa$
and all quantum states
  $\sigma$, $(\sigma; \vec{q}; P\kappa) \pbsim (\sigma; \vec{q}; Q\kappa)$. 
\end{Definition}

Importantly, the equivalence is a \emph{congruence}, which means it is preserved in all contexts.

\begin{Theorem}[Full probabilistic branching bisimilarity is a congruence]
  \label{thm:congruence}
  If $P \fpbsim Q$ then for any context $\ctxt{C}{}$, if $\ctxt{C}{P}$ and $\ctxt{C}{Q}$ are typable then $\ctxt{C}{P} \fpbsim \ctxt{C}{Q}$. 
\end{Theorem}

\subsection{Correctness of $Teleport$}

We now sketch the proof that $\pteleport \fpbsim \pqwire$, which by Theorem~\ref{thm:congruence} implies that this works in any context.
works in any context. 
\begin{proposition}
  $\pteleport \fpbsim \pqwire$.
\end{proposition}
 \begin{proof}
 First we prove that $\pteleport \pbsim \pqwire$, by defining an equivalence relation $\mathcal{R}$ that contains the pair $(\cnfig{\sigma}{\emptyset}{\pteleport},\cnfig{\sigma}{\emptyset}{\pqwire})$ for all $\sigma$ and is closed under their transitions. $\mathcal{R}$ is defined by taking its equivalence classes to be the $F_i(\sigma)$ defined below, for all states $\sigma$, which group configurations according to the sequences of observable transitions leading to them.
 \[
 \begin{array}{rcl}
 \pname{F_{1}(\sigma)} & = & \{f \mid
 (\sigma;\emptyset;P)\weaktrans{}f ~\mbox{and}~ P \in E\}\\
 \pname{F_{2}(\sigma)} & = & \{f \mid
 (\sigma;\emptyset;P)\weaktrans{\inp{c}{q_{1}}}f ~\mbox{and}~ P \in E\}\\
 \pname{F_{3}(\sigma)} & = & \{f  \mid  (\sigma;\emptyset;P)
 \weaktrans{\inp{a}{q_{1}}} \weaktrans{\outp{d}{q_{2}}}f ~\mbox{and}~ P \in E\}
 \end{array}
 \]
Here $E$ is $\{\pteleport, \pqwire\}$ and we now prove that $\mathcal{R}$ is a probabilistic branching
 bisimulation. It suffices to consider transitions between $F_i$
 classes, as transitions within classes must be $\tau$ and are matched
 by $\tau$.  If $f, g \in F_{1}(\sigma)$ and $f\transition{\tau}f'$ then $f'\in F_{1}(\sigma)$ and therefore $(f',g) \in \mathcal{R}$. Otherwise if $f\transition{\inp{c}{q_1}}f'$ then $f'\in F_{2}(\sigma)$ and we find $g',g''$ such that $g\weaktrans{}g'\transition{\inp{c}{q_{1}}}g''$ with $g'\in F_{1}(\sigma)$ and $g''\in F_{2}(\sigma)$, so $(f,g')\in\mathcal{R}$ and $(f',g'')\in\mathcal{R}$ as required. 
 
 If $f, g \in F_{2}(\sigma)$ and $f\transition{\tau}f'$ then $f'\in F_{2}(\sigma)$ and therefore $(f',g) \in \mathcal{R}$. Otherwise if $f\transition{\outp{d}{q_2}}f'$ then $f'\in F_{3}(\sigma)$ and we find $g',g''$ such that $g\weaktrans{}g'\transition{\outp{d}{q_{1}}}g''$ with $g'\in F_{2}(\sigma)$ and $g''\in F_{3}(\sigma)$.

If $f$ is a mixed configuration arising from $\pteleport$ then for any arbitrary state $\sigma$, with reference to the execution of the protocol discussed in \cite{Gay2013}, we have the same reduced density matrices for both processes.  \qed
\end{proof}

\section{Conclusion and Future Work}
\label{sec-conclusion}
\label{sec:conclusion}
We have presented the extensions of the language and theories of CQP to analyse higher dimensional quantum protocols like qudit teleportation. Investigations are in progress to apply this analysis to other protocols like superdense coding protocol and secret sharing. We believe this should be completed very soon. Previous work on CQP has established the possibility of modelling and analysing a physical realisation of quantum computing systems \cite{Arnold2013,Arnold2014} such as linear optical quantum computing. This is achieved by defining the physical linear optical elements in CQP, and then analysing the optical gates that are used for quantum information processing. In the current line of work, the next important future task is to extend these CQP definition to model realistic higher dimensional quantum systems, that is, systems that exploit the orbital angular momentum of light. The advantage of using process calculus approach is that it provides a systematic methodology for verification of quantum systems. Significantly, the equivalence defined in CQP is a congruence, meaning that equivalent processes remain equivalent in any context, and supporting equational reasoning \cite{Gay2015}. Following the established work in process calculi, one of the long-term goal in CQP is to develop software for automated analysis of quantum information processing systems.

\bibliographystyle{splncs}

\bibliography{report_main}

\section{Appendix}
\label{sec:Semantics}
\label{sec-Semantics}
\subsection{Syntax}
\begin{figure*}[h!]
  \begin{eqnarray*}
    T & \bnf & \Int \alt \Qbit \alt \NS \alt \chant{\tilde{T}} \alt \Op{1} \alt \Op{2} \alt \cdots \\
    v & \bnf & \mktype{0} \alt \mktype{1} \alt \cdots \alt \qgate{H} \alt \cdots \\
    e & \bnf & v \alt x \alt \measure{\tilde{e}} \alt {\tilde{e}}\trans{e^e} \alt e+e \\
    P & \bnf & \nil \alt (P | P) \alt P + P \alt \inp{e}{\tilde{x}:\tilde{T}}.P \alt  \outp{e}{\tilde{e}}.P \alt \{e\}.P \alt [e].P \alt (\qdit x)P \alt  (\new x:\chant{T})P 
  \end{eqnarray*}
  \caption{\label{fig:cqp_syntax}Syntax of CQP.}
\end{figure*}
The behaviour of the processes is precisely specified by the formal semantics of CQP. Full details can be found in Puthoor's thesis \cite{PuthoorThesis}. 
\begin{figure*}[t]
  \begin{eqnarray*}
    v & \bnf & \ldots \alt q \alt c \\
    E & \bnf & [] \alt \measure{E,\tilde{e}} \alt \measure{v, E, \tilde{e}} \alt \dots \alt \measure{\tilde{v}, E} \alt E + e \alt v + E \\
    F & \bnf & \inp{[]}{\tilde{x}}.P \alt \outp{[]}{\tilde{e}}.P \alt \outp{v}{[].\tilde{e}}.P \alt \outp{v}{v,[],\tilde{e}}.P \alt \cdots \alt \outp{v}{\tilde{v},[]}.P \alt \{[]\}.P
  \end{eqnarray*}
  \caption{\label{fig:cqp_internal_syntax}Internal syntax of CQP.}
\end{figure*}
The syntax of CQP is defined by the grammar as shown in Figure \ref{fig:cqp_syntax}. We use the notation $\tilde{e} = e_1,\ldots,e_n$, and write $|\tilde{e}|$ for the length of a
tuple. The syntax consists of types $T$, values $v$, expressions $e$ (including quantum measurements and the conditional application of unitary operators $\tilde{e}\trans{e^{e}}$), and processes $P$. Values $v$ consist of variables ($x$,$y$,$z$ etc), literal values of data types (0,1,..), unitary operators such as the Hadamard operator $\mathsf{H}$. Expressions $e$ consist of values, measurements $\measure e_{1},\dots,e_{n}$, applications $e_{1},\dots,e_{n} \trans e$ of unitary operators, and expressions involving data operators such as $e + e'$. Processes include the nil process $\nil$, parallel composition $P|P$, inputs $\inp{e}{\tid{\vec{x}}{\vec{T}}}.P$, outputs $\outp{e}{\vec{e}}.P$, actions $\{e\}.P$ (typically a unitary operation or measurement), typed channel restriction $(\new x: \chant{\vec{T}})P$ and qudit declaration $(\qdit x)P$. In order to define the operational semantics we provide the \emph{internal syntax} in Figure \ref{fig:cqp_internal_syntax}. Values are supplemented with qudit names $q$, which are generated at run-time and substituted for the variables used in $\qdit$. Evaluation contexts for expressions ($E[]$) and processes ($F[]$) are used to define the operational semantics \cite{Wright1994}.

\subsection{Semantics}

In CQP,  the execution of a system is not completely described by the process term (which is the case for classical process calculus) but also depends on the quantum state. Hence the operational semantics are defined using \emph{configurations}, which represent both the quantum state and the process term. 
A \emph{configuration} is a tuple ($\sigma;\omega;P$) where $\sigma$ is a mapping from qudit names to the quantum state, $\omega$ is a list of qudit names, and P is a process. 
For example, we have a configurations such as
\begin{equation}\label{eq:quantum-output}
\cnfig{\qstore{q,r}{\frac{1}{\sqrt{d}}\sum_{j=0}^{d-1}\ket{j}_{q}\otimes\ket{j}_{r}}}{q}{\outp{c}{q}\sep P}.
\end{equation}
This configuration means that the global quantum state consists of two qudits, $q$ and $r$, in the specified state; that the process term under consideration has only access to qudit $q$; and that the process itself is $\outp{c}{q}\sep P$. Now consider a configuration with the same quantum state but a different process term:
\[
\cnfig{\qstore{q,r}{\frac{1}{\sqrt{d}}\sum_{j=0}^{d-1}\ket{j}_{q}\otimes\ket{j}_{r}}}{r}{\outp{b}{r}\sep Q}.
\]
The parallel composition of these configurations is the following:
\[
\cnfig{\qstore{q,r}{\frac{1}{\sqrt{d}}\sum_{j=0}^{d-1}\ket{j}_{q}\otimes\ket{j}_{r}}}{q,r}{\outp{c}{q}\sep P
\parallel \outp{b}{r}\sep Q}
\]
where the quantum state is still the same.  Just like the classical process calculus, we have the semantics of CQP defined in terms of labelled transitions between the configurations.  For example, configuration (\ref{eq:quantum-output})
has the transition
\[
\ptrns{\cnfig{\qstore{q,r}{\frac{1}{\sqrt{d}}\sum_{j=0}^{d-1}\ket{j}_{q}\otimes\ket{j}_{r}}}{q}{\outp{c}{q}\sep
  P}}{\outp{c}{q}}{\cnfig{\qstore{q,r}{\frac{1}{\sqrt{d}}\sum_{j=0}^{d-1}\ket{j}_{q}\otimes\ket{j}_{r}}}{\emptyset}{P}}.
\]
The quantum state is not changed by this transition, but because qudit $q$ is output, the continuation process $P$ no longer has access to it; the final configuration has an empty list of owned qudits. 

Importantly, we should note that in CQP, we get a probability distribution over configurations after a measurement, which then later reduces probabilistically to one particular configuration. To prove that the equivalence of CQP, a more sophisticated analysis of measurement in the semantics called the \emph{mixed configurations} is included. We define a \emph{mixed configuration} as a weighted distribution over pure configurations. 
A mixed configuration is a weighted distribution, written as 
\[
\begin{array}{l}
\Dist{i \in I}{g_i} ( \qstore{q}{\ket{\psi_{i}}} ; \omega; \ltrm{x}{P}{\tilde{v_{i}}} ) 
\end{array}
\]
with weights $g_{i}$ where $\sum_{i\in I}g_{i} = 1$ and for each $i \in I, 0 < g_{i} \le 1$ and $\ket{\psi_{i}} \in \mathbb{H}^{2|\tilde{q}|}$ and $|\tilde{v_{i}}| = |\tilde{x}|$. The operator $\oplus$ represents a distribution over the index set $I$ with weights $g_{i}$. The process term is replaced by the expression $\lambda\tilde{x}.P;\tilde{v_{i}}$ which indicates that in each component the variables $\tilde{x}$, appearing in $P$ as placeholders, should be substituted for the values $\tilde{v_{i}}$. A pure configuration can be considered as a mixed configurations with a single component. If the observer does not get the result of a quantum measurement then we say that the system is in a mixed state.

A few of the important transition relations for evaluating values and expressions are defined by the rules in Figure~\ref{fig:trans_exp}. The complete details are provided in \cite{PuthoorThesis,Gay2013}
\begin{figure*}
  \begin{gather*}
    \tag{\Rplus}
    ( \qstore{\vec{q}}{\ket{\psi}} ; \omega; u + v) \transitionv ( \qstore{\vec{q}}{\ket{\psi}} ; \omega; \ltrm{x}{x}{w}) \quad \textrm{where $w = u + v$} \\
    \tag{\Rmeasure}
    \begin{array}{r}
      ( \qstore{q_0,\dots,q_{n-1}}{\alpha_0 \ket{\phi_0} + \dots + \alpha_{d^n-1} \ket{ \phi_{d^n-1}}}; \omega; \measure{ q_0,\dots, q_{r-1} } ) \transitionv \\
      \Dist{0 \le m < d^r}{g_m} ( \qstore{q_0, \dots, q_{n-1}}{\frac{\alpha_{l_m}}{\sqrt{g_m}} \ket{ \phi_{l_m}} + \dots + \frac{ \alpha_{u_m}}{\sqrt{g_m}} \ket{ \phi_{u_m}}}; \omega; \ltrm{x}{x}{m} )
    \end{array} \\
    \textrm{where } l_m = d^{n-r}m, u_m = d^{n-r}(m+1)-1, g_m = |\alpha_{l_m}|^2 + \dots + |\alpha_{u_m}|^2 \\
    \tag{\Rtrans}
    ( \qstore{q_0,\dots,q_{n-1}}{\ket{\phi}} ; \omega; q_0,\dots,q_{r-1}\trans{U^m} ) \transitionv \hspace{20mm} \\
    \hspace{30mm} (\qstore{q_0,\dots,q_{n-1}}{(U^m \otimes I_{n-r}) \ket{\phi}}; \omega; \unit; \cdot ) \\
    \tag{\Rcontext}
    \begin{prooftree}
      \forall i \in I. ( \qstore{\vec{q}}{\ket{\psi_i}} ; \omega; e\{\vec{u}_i/\vec{y}\}) \transitionv \Dist{j \in J_i}{g_{ij}} ( \qstore{\vec{q}}{\ket{\psi_{ij}}} ; \omega; \ltrm{\vec{x}}{e'\{\vec{u}_i/\vec{y}\}}{\vec{v}_{ij}} )
      \justifies
      \Dist{i \in I}{h_i} ( \qstore{\vec{q}}{\ket{\psi_i}} ; \omega; \ltrm{\vec{y}}{E[e]}{\vec{u}_i}) \transitione \Dist{\substack{i \in I\\j \in J_i}}{h_ig_{ij}} ( \qstore{\vec{q}}{\ket{\psi_{ij}}} ; \omega; \ltrm{\vec{y}\vec{x}}{E[e']}{\vec{u}_i,\vec{v}_{ij}})
    \end{prooftree}
  \end{gather*}
  \caption{Transition rules for values and expressions.}
\label{fig:trans_exp}
\end{figure*}

\end{document}